\documentclass[letterpaper,10pt,conference]{ieeeconf}

\IEEEoverridecommandlockouts 
\usepackage{amsmath,amssymb,amsfonts}
\usepackage{amsthm}
\usepackage{graphicx}
\usepackage{textcomp}
\usepackage[dvipsnames]{xcolor}
\usepackage{booktabs}
\usepackage{cite}
\usepackage[hidelinks]{hyperref}
\hypersetup{pdftitle={SparseMax-dSTL: Sound Robustness with Horizon-Independent Gradient Allocation for Temporal Logic Trajectory Optimization}}
\usepackage[capitalise,noabbrev]{cleveref}
\usepackage{url}
\usepackage[font=footnotesize,labelfont=bf]{caption}
\usepackage{macros}

\newif\ifarxiv
\arxivfalse
\ifdefined\buildarxiv\arxivtrue\fi

\graphicspath{{figures/}}

\title{\LARGE\bfseries SparseMax-\texorpdfstring{$\partial$}{d}STL: Sound Robustness with Horizon-Independent Gradient Allocation for Temporal Logic Trajectory Optimization}

\author{Sai Sandeep Damera, Ryan Matheu, John S. Baras and Calin Belta%
\thanks{The authors are with the University of Maryland, College Park, USA. Emails: \{\texttt{sdamera}, \texttt{rmatheu}, \texttt{baras}, \texttt{calin}\}\texttt{@umd.edu}. Code available at \url{https://github.com/sdamera95/SparseMax-dSTL}.}%
}

\begin{document}
\maketitle
\thispagestyle{empty}
\pagestyle{empty}

\begin{abstract}
A gradient-based solver that enforces a Signal Temporal Logic (STL) specification uses a smooth approximation of its quantitative semantics, the robustness. In safety-critical trajectory optimization, this approximation must have two properties. First, it must be sound, that is, nonnegative only on trajectories that satisfy the specification. Second, at each temporal operator, the weight of its gradient on the violating samples must stay bounded away from zero, however long the horizon. The approximations in use, including the log-sum-exp, its sound variant, and the generalized-mean robustness, do not have both. We introduce lower approximations of the minimum and the maximum whose gradients are sparsemax projections, and we prove that the robustness built from them has both. It lies below the exact robustness and is therefore sound, and its error has a bound that is independent of the number of samples and known before optimization. The gradient of each approximation is exactly zero on every sample beyond a data-dependent threshold, so the weight on violating samples a gap beyond the others has a positive lower bound at every horizon. On a planar unicycle and on a seven-joint manipulator beside a person in simulation, the log-sum-exp approximation is positive on a violating trajectory. As the horizon of the manipulator's task grows from 10 to 80 s, the gradient weight of the sound variant on the violating samples falls from 0.22 to 0.04, while that of sparsemax stays at 1. In a constrained trajectory optimization of the unicycle, the constraint on the log-sum-exp robustness returns a violating trajectory, whereas the constraint on the sparsemax robustness returns a trajectory that satisfies the specification with the required margin.
\end{abstract}

\section{Introduction}
\label{smd:sec:introduction}
A gradient-based solver that enforces a signal temporal logic (STL) specification as a constraint does not evaluate the specification itself. STL assigns trajectories a scalar robustness, a nesting of minima and maxima over predicates and time samples whose sign determines satisfaction~\cite{maler2004monitoring,donze2010robust}. The nesting is nonsmooth, so the solver differentiates a smooth surrogate of every extremum~\cite{leung2023backpropagation,kapoor2025stlcg++} and accepts a trajectory as feasible when the resulting smoothed robustness is nonnegative. The surrogate therefore determines which trajectories the solver accepts as feasible and how its gradient is allocated among the time samples.

Consider a robot manipulator that hands a part to a person. Its STL specification requires that it stay out of the zone into which the person reaches until it has picked the part, keep a minimum separation from the person, and move slowly near the person. The first requirement is an until operator, whose left operand (stay out of the zone) must hold until its right operand (the pick) holds at some sample of a time window. Each sample of the window is a witness, and the robustness of the operator is a maximum over its witnesses. A trajectory that crosses the zone for a third of a second before the pick violates the specification at fewer than twenty of the hundreds or thousands of samples in the horizon, while the manipulator waits at the zone boundary for the rest. \Cref{smd:fig:geometry}(a) sketches such a trajectory with nine samples. A smoothed robustness that constrains such a trajectory must have two properties. First, it must be sound, that is, negative on every violating trajectory. Second, its gradient allocation, the weight that its gradient places on each sample, must not vanish on the violating samples as the waiting time grows. A policy gradient or a gradient-based safety filter that differentiates the robustness needs the same two properties.

\begin{figure}[t]
\centering
\includegraphics{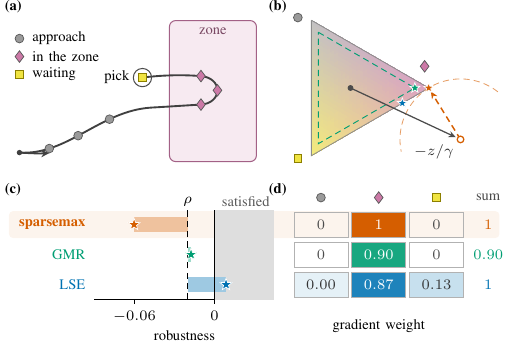}
\caption{(a) \emph{Stay out of the zone until the pick} on nine samples, three per class (approach, in the zone, waiting), with margins 0.25, $-0.02$, 0.04 from the zone and $-1$, 0.8, 0.8 at the pick, $\varepsilon=0.07$ per node. (b) The prefix minimum on the simplex of the class weights (illustrative, not to scale): the values $z$ push the centre (equal weights) by $-z/\gamma$, and the nearest point of the simplex gives the sparsemax weights; the LSE weights lie inside, the GMR weights sum to $0.90$ (dashed). (c) Values against the exact $\rho$ and (d) weight per class at the prefix minimum of (b); GMR of order $(-10,10)$.}
\label{smd:fig:geometry}
\end{figure}

The smoothing techniques in use today provide at most the first. The log-sum-exp (LSE) robustness~\cite{pant2017smooth} replaces the maximum of $m$ sample values $z=(z_1,\ldots,z_m)$ by $\beta^{-1}\log\sum_ie^{\beta z_i}$ with a parameter $\beta>0$. It exceeds the maximum by up to $\log m/\beta$~\cite{welikala2023smooth}, attained when all values are equal, and it can be nonnegative on a violating trajectory~\cite{gilpin2020smooth}. We call this smoothing the plain LSE. We show that an until operator produces such ties when the deepest violation of its left operand precedes the window, so the plain LSE robustness is positive on a shallow violation. Sound measures remove that defect. The softmax-mean measure under-approximates both extrema~\cite{gilpin2020smooth}, and the arithmetic-geometric mean and the generalized-mean robustness (GMR) preserve the sign of the exact robustness~\cite{mehdipour2019arithmetic,mehdipour2025generalized}. We ask how much weight a temporal operator keeps on its violating samples under each sound measure as the horizon adds samples. Subtracting from the LSE its largest excess, $\log m/\beta$, gives the sound LSE, which never exceeds the maximum and whose gradient, the softmax of $\beta z$, gives every sample a positive weight. At a worst-case error $\varepsilon$ per operator, its weight on the violating samples tends to zero as the horizon adds samples whose values lie within $\varepsilon$ of theirs. By eq.~(12) of~\cite{mehdipour2025generalized}, GMR averages the violating samples over all $d$ samples of the operator, so its weight on $k$ of them is of order $(k/d)^{1/q}$ for the exponent $q$ of its power mean.

We replace every extremum by a lower approximation derived from the maximum regularized by a quadratic penalty of weight $\gamma>0$, whose gradient is the sparsemax projection onto the simplex~\cite{martins2016Sparsemax,mensch2018differentiable}. \Cref{smd:fig:geometry} shows the construction on these nine samples. The value lies below the exact extremum by at most $\gamma/2$ for any number of samples, and every sample farther than a data-dependent threshold from the extremum receives a derivative of exactly zero. To our knowledge, the resulting robustness is the first smooth one that is a lower bound of the exact robustness and whose weight on the violating samples, for every gap by which they lie beyond the others, has a lower bound that does not depend on the horizon. Our contributions are the following.

\noindent\textbf{1)~Sound robustness from sparsemax lower extrema.} We construct lower approximations of both extrema from the regularized maximum of~\cite{mensch2018differentiable} and prove that the robustness built from them, until included, is a lower bound of the exact one with an error known before optimization (\cref{smd:prop:bands,smd:thm:composition}, composed as in~\cite{welikala2023smooth}). We also bound the excess of the plain LSE at an until operator from below and obtain the depth of violation up to which its robustness is positive (\cref{smd:prop:flip}).

\noindent\textbf{2)~Horizon-independent gradient allocation.} We prove that the sparsemax weight on $k$ violating samples that lie a gap $\Delta$ beyond the $N$ others is at least $\min\{1,k\Delta/\gamma\}$ for every $N$. At equal worst-case error $\varepsilon$, the weight of the sound LSE tends to zero as $N$ grows if $\Delta<\varepsilon$ (\cref{smd:thm:allocation,smd:cor:gap}).

\noindent\textbf{3)~Numerical examples.} We show on violating trajectories of a planar unicycle and a manipulator that the plain LSE robustness is positive and that the measured weights follow the closed forms. The optimizations demonstrate the two properties with a first-order solver, where the gradient's allocation sets the update; the immediate update advantage is claimed for first-order methods only (\cref{smd:sec:results}).

\ifarxiv Implementation details and further measurements are in the appendices.\fi

\section{Related Work}
\label{smd:sec:related}
STL robustness reduces Boolean and temporal requirements to nested extrema over predicate values~\cite{donze2010robust}. Differentiable STL computation graphs expose these quantities to gradient-based methods~\cite{leung2023backpropagation}, and temporal masking parallelizes their evaluation~\cite{kapoor2025stlcg++}. We keep such a graph and study the functions placed at its extremum nodes. An LSE maximum in such a graph is an upper bound of the maximum, so the resulting robustness is not a lower bound, and \cref{smd:prop:flip} shows where the excess is attained.

The softmax-mean measure pairs an LSE minimum with a softmax-weighted mean as the maximum, both under-approximations, so a positive value implies satisfaction for every parameter value~\cite{gilpin2020smooth}. Its minimum is the LSE minimum, so \cref{smd:thm:allocation} applies to it at minimum nodes, where the prefix of an until is evaluated. Average-based measures replace the extrema by means, first the arithmetic-geometric mean~\cite{mehdipour2019arithmetic} and then power means and generalized functional means~\cite{mehdipour2025generalized}. Their purpose is to remove the locality and masking of the extrema, by which the robustness depends on one operand at one time, so that every subformula and every time sample contributes to the value. They obtain soundness by branching on the sign of the operands, so the sign of the measure equals the sign of the exact robustness. Their conjunction lies at or above the minimum, so these measures preserve the sign and not a margin, and their distance from the exact robustness depends on the operands. A generalized-mean measure that is also continuously differentiable follows the same design~\cite{uzun2024optimization}. A smooth, idempotent, and commutative conjunction has the derivative $1/M$ with respect to each of $M$ equal nonzero operands~\cite{varnai2020robustness}. A comparison of four smooth measures composes per-node error bounds through the specification by the recursion we also use and allows different operator parameters at different levels of the specification~\cite{welikala2023smooth}.

Sparsemax is the Euclidean projection onto the probability simplex~\cite{martins2016Sparsemax}, and regularized maxima with their use in differentiable dynamic programming supply the operator we build on~\cite{mensch2018differentiable}. The binary polynomial smooth minimum in~\cite[Section~5.2.4]{mao2021successive} coincides with $Q^-_\gamma$ of \cref{smd:eq:lower} for two entries. Our contribution is the lower-bound construction for both extrema at any number of entries, its composition through STL, and the allocation result at equal error.

Our operators stand between the extrema and the means. Every sample within a threshold of the extremum contributes, so tied or nearly tied violating samples share the derivative as they do under a mean, and every sample beyond the threshold receives none. It is observed in~\cite{mehdipour2025generalized} that a measure which considers only the most critical time suits control in safety-critical applications and that the exponents of the power means set how much the other times contribute. The parameter $\gamma$ has that role here, with two guarantees: a lower bound of the exact robustness with an error known before optimization, where a mean preserves the sign alone, and a weight on the violating samples bounded below independently of the number of other samples (\cref{smd:cor:gap}). To our knowledge no earlier smooth measure provides both, and none of the works above gives the weight on the violating samples beyond a tie.

\section{Problem Formulation}
\label{smd:sec:problem}
\subsection{Specifications and exact robustness}
Let $x_{t+1}=f_h(x_t,u_t)$ be a discrete-time robot model with sampling period $h>0$ and fixed initial state $x_0$. The controls $U=(u_0,\ldots,u_{T-1})$ lie in a box $\mathcal U$ of actuator bounds and generate the sampled trajectory $X(U)=(x_0,\ldots,x_T)$ over the horizon $H=Th$. An STL specification in negation normal form is built from the grammar
\begin{equation}
\phi::=\mu\mid\neg\mu\mid\phi\wedge\psi\mid\phi\vee\psi\mid
\mathbf G_I\phi\mid\mathbf F_I\phi\mid\phi\,\mathbf U_I\,\psi,
\label{smd:eq:grammar}
\end{equation}
where $\mu$ is a predicate with a continuously differentiable value $g_\mu(x_t)$ that is nonnegative when the predicate holds, so that $\rho^\mu_t=g_\mu(x_t)$ and $\rho^{\neg\mu}_t=-g_\mu(x_t)$, and $I=\{a,\ldots,b\}$ with integers $0\le a\le b$ is a window of sample offsets. The robustness of a specification at sample $t$ is defined recursively from the predicate values by
\begin{equation}
\begin{aligned}
\rho^{\phi\wedge\psi}_t&=\min\{\rho^\phi_t,\rho^\psi_t\},\qquad
\rho^{\phi\vee\psi}_t=\max\{\rho^\phi_t,\rho^\psi_t\},\\
\rho^{\mathbf G_I\phi}_t&=\min_{j\in I}\rho^\phi_{t+j},\qquad
\rho^{\mathbf F_I\phi}_t=\max_{j\in I}\rho^\phi_{t+j},\\
\rho^{\phi\,\mathbf U_I\,\psi}_t&=\max_{j\in I}\,\min\Bigl\{\rho^\psi_{t+j},\ \min_{0\le i\le j}\rho^\phi_{t+i}\Bigr\},
\end{aligned}
\label{smd:eq:semantics}
\end{equation}
and we say that the sampled trajectory satisfies $\phi$ when $\rho(X)=\rho^\phi_0\ge0$. We use the until semantics with a closed prefix~\cite{sadraddini2015robust,mehdipour2025generalized}, in which the witness at offset $j\in I$ requires $\psi$ at $t+j$ and $\phi$ at every sample from $t$ through $t+j$. Unrolling \cref{smd:eq:semantics} over the samples gives a finite acyclic evaluation graph whose leaves are predicate values and whose internal nodes are minima and maxima. A node is evaluated only where every sample its window reads exists, so the horizon of a specification is at most $H$. We write $\mathrm{ch}(v)$ for the children of node $v$ and $m_v$ for their number. \Cref{smd:fig:until} shows the two levels of an until node.

\begin{figure}[t]
\centering
\includegraphics{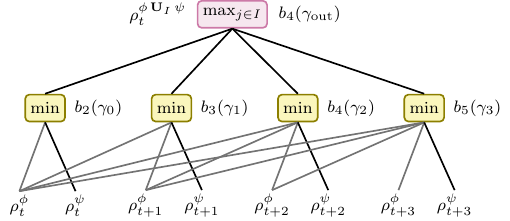}
\caption{The until node $\phi\,\mathbf U_I\,\psi$ at sample $t$ for a window of four offsets. The outer maximum ranges over the witnesses $j\in I$, and the child of witness $j$ is a minimum over $\rho^\psi_{t+j}$ and the closed prefix $\rho^\phi_t,\ldots,\rho^\phi_{t+j}$. Each node is labeled with its error band $b_m(\gamma)$ of \cref{smd:prop:bands}.}
\label{smd:fig:until}
\end{figure}

\subsection{Constrained trajectory optimization}
We seek controls of least effort $E(U)$ that satisfy a specification $\phi=\bigwedge_{j=1}^{J}\phi_j$,
\begin{equation}
\min_{U\in\mathcal U}\ E(U)
\quad\text{subject to}\quad
\rho^{\phi_j}_0(X(U))\ge0,\quad j=1,\ldots,J.
\label{smd:eq:exact}
\end{equation}
The conjunction at the top of the specification is imposed exactly as $J$ constraints. Since each $\rho^{\phi_j}_0$ is nonsmooth in $U$, a gradient-based solver replaces every extremum node of the graph by a smooth function of its children and solves
\begin{equation}
\min_{U\in\mathcal U}E(U)
\ \ \text{subject to}
\ \ \widetilde\rho_j(X(U))\ge c,\ \ j=1,\ldots,J,
\label{smd:eq:conjuncts}
\end{equation}
with a margin $c\ge0$. The smoothing thus sets the feasible set of \cref{smd:eq:conjuncts} and the search direction of the solver.

\subsection{Two requirements on the smoothing}
The first requirement concerns the feasible set.
\begin{definition}
\label{smd:def:sound}
A smoothed robustness $\widetilde\rho$ is \emph{sound} if $\widetilde\rho(X)\ge0$ implies $\rho(X)\ge0$ for every trajectory $X$. It is a \emph{lower bound with budget} $B\ge0$ if $\rho(X)-B\le\widetilde\rho(X)\le\rho(X)$ for every $X$.
\end{definition}
A lower bound is sound, and a margin on it is a margin on the exact robustness, since $\widetilde\rho\ge c$ implies $\rho\ge c$. With lower bounds $\widetilde\rho_j$ the feasible set of \cref{smd:eq:conjuncts} lies inside that of \cref{smd:eq:exact}, every feasible trajectory satisfies the specification, and the solver may return the first one it reaches. If $\widetilde\rho_j\ge c$ can hold where $\rho^{\phi_j}_0<0$, the feasible set contains violating trajectories, which a solver that minimizes the effort can return.

The second requirement concerns the search direction. By the chain rule, the gradient of a smoothed robustness with respect to the controls is a weighted sum $\sum_\ell\pi_\ell\nabla_Ug_\ell(X(U))$ over the predicate leaves $\ell$. The leaf weight $\pi_\ell$ is the sum over root-to-leaf paths of the products of the derivatives of the smooth extrema along the path. At one extremum node with child values $z\in\R^m$, exact extremum $A$ and smooth replacement $\widetilde A$, we call $\sum_{i\in C}\partial\widetilde A(z)/\partial z_i$ the weight on the children $C$. On a violating trajectory the children that attain the minimum of a temporal operator over predicates are violating samples, and the number of its other children grows with the horizon.

\begin{problem}
\label{smd:prob:main}
Find smooth replacements of the minimum and the maximum such that, for every specification of \cref{smd:eq:grammar},
\begin{enumerate}
\item[(R1)] the smoothed robustness is a lower bound whose budget depends only on the evaluation graph and the smoothing parameters, and
\item[(R2)] at a node whose $k$ extreme children are separated from its $N$ other children by a gap $\Delta>0$, the weight on the $k$ children stays bounded away from zero as $N$ grows.
\end{enumerate}
\end{problem}
Every smooth extremum tends to the exact one as its worst-case error $\sup_z|\widetilde A(z)-A(z)|$ tends to zero, and its gradient then changes more abruptly with the trajectory. We therefore pose (R2) at a fixed worst-case error $\varepsilon$ per node and compare two smoothings at equal $\varepsilon$.

\section{Sparsemax Lower Extrema}
\label{smd:sec:method}
\subsection{The operators}
For $z\in\R^m$ and $\gamma>0$, the quadratically regularized maximum over the simplex $\mathcal S_m=\{p\in\R^m:p\ge0,\ \1^\top p=1\}$ and its maximizer are
\begin{equation}
M_\gamma(z)=\max_{p\in\mathcal S_m}\Bigl(p^\top z-\tfrac\gamma2\|p\|_2^2\Bigr),\quad
p_\gamma(z)=\Pi_{\mathcal S_m}(z/\gamma).
\label{smd:eq:regularized}
\end{equation}
The objective is strictly concave, completing the square gives the Euclidean projection $\Pi_{\mathcal S_m}$, which is the sparsemax map~\cite{martins2016Sparsemax}, and the envelope theorem gives $\nabla M_\gamma(z)=p_\gamma(z)$. The optimality conditions of the projection are
\begin{equation}
p_{\gamma,i}(z)=\frac{(z_i-\theta)_+}{\gamma},
\qquad \sum_{i=1}^{m}(z_i-\theta)_+=\gamma,
\label{smd:eq:threshold}
\end{equation}
so every entry at or below the threshold $\theta$ receives a derivative of exactly zero (\cref{smd:fig:geometry}(b)). The maximum of the penalty over the simplex is $\gamma/2$ and its minimum is $\gamma/(2m)$, so $M_\gamma(z)$ lies between $\max_iz_i-\gamma/2$ and $\max_iz_i-\gamma/(2m)$~\cite[Lemma~1]{mensch2018differentiable}. Shifting by these two constants gives lower approximations of both extrema,
\begin{equation}
Q^+_\gamma(z)=M_\gamma(z)+\frac\gamma{2m},\qquad
Q^-_\gamma(z)=-M_\gamma(-z)-\frac\gamma2.
\label{smd:eq:lower}
\end{equation}

\begin{prop}[Error bands]
\label{smd:prop:bands}
Let $b_m(\gamma)=\frac\gamma2(1-1/m)$. For every $z\in\R^m$,
\begin{equation}
\begin{aligned}
0&\le\max_i z_i-Q^+_\gamma(z)\le b_m(\gamma),\\
0&\le\min_i z_i-Q^-_\gamma(z)\le b_m(\gamma).
\end{aligned}
\label{smd:eq:bands}
\end{equation}
Both operators are continuously differentiable and nondecreasing in every entry, with gradients $p_\gamma(z)$ and $p_\gamma(-z)$ in the simplex.
\end{prop}
\begin{proof}
A vertex that attains $\max_iz_i$ gives $M_\gamma(z)\ge\max_iz_i-\gamma/2$. Every $p\in\mathcal S_m$ satisfies $p^\top z\le\max_iz_i$ and $\|p\|_2^2\ge1/m$, so $M_\gamma(z)\le\max_iz_i-\gamma/(2m)$. Applying both bounds to $z$ and to $-z$ gives \cref{smd:eq:bands}. The gradients follow from the chain rule, and a gradient in the simplex has nonnegative entries.
\end{proof}

The band does not depend on $z$ and is at most $\gamma/2$ for every $m$, at minimum and maximum nodes alike. \Cref{smd:fig:operator}(a) shows $Q^-_\gamma$ inside its band.

\begin{figure}[t]
\centering
\includegraphics{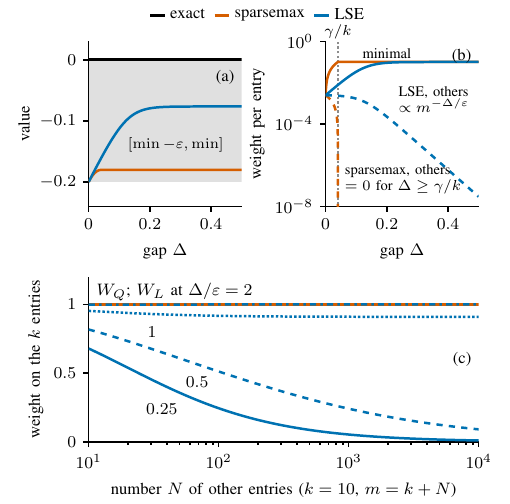}
\caption{The minima of sparsemax and of the sound LSE on $k=10$ entries at the minimum and $N$ entries a gap $\Delta$ above it, at equal worst-case error $\varepsilon=0.2$; $N=400$ in (a) and (b). (a) Both values lie in the band $[\min z-\varepsilon,\min z]$. (b) Weight on one minimal entry (solid) and on one other entry (dashed), exactly zero for sparsemax once $\Delta\ge\gamma/k$. (c) Weight on the $k$ entries against $N$ (\cref{smd:thm:allocation}): $W_Q=1$ at every ratio $\Delta/\varepsilon$ shown, since $k\Delta\ge\gamma$ there, and each $W_L$ curve is labeled by its ratio.}
\label{smd:fig:operator}
\end{figure}

\subsection{A lower bound on the evaluation graph}
Replace the extremum at every internal node $v$ by the operator of \cref{smd:eq:lower} with the number $m_v$ of its children and its parameter $\gamma_v$, and let $\widetilde\rho_v$ and $\rho_v$ be the smoothed and exact values of the node. A predicate leaf may itself be evaluated by a smooth function that lies at most $e_v\ge0$ below the exact predicate and never above it. Define the budget
\begin{equation}
B_v=\begin{cases}
e_v,&\text{leaf},\\
b_{m_v}(\gamma_v)+\max_{u\in\mathrm{ch}(v)}B_u,&\text{otherwise},
\end{cases}
\label{smd:eq:budget}
\end{equation}
which is the recursion that composes per-node bands of LSE and softmax operators in~\cite{welikala2023smooth}. Every node is evaluated as one operator over all of its children, as \cref{smd:fig:until} draws the two levels of an until; the operators are not associative, so splitting a node into nested ones changes the function and adds a band to its path.

\begin{thm}[Lower bound]
\label{smd:thm:composition}
On a finite acyclic evaluation graph of minima and maxima whose leaves satisfy $0\le\rho_v-\widetilde\rho_v\le e_v$, every node satisfies
\begin{equation}
0\le\rho_v-\widetilde\rho_v\le B_v.
\label{smd:eq:composition}
\end{equation}
At the root, $B_v$ is the largest sum, along a root-to-leaf path, of the bands of its nodes and the error of its leaf.
\end{thm}
\begin{proof}
The claim holds at the leaves by assumption. Suppose it holds at every child of a node, and let $a$ and $\widetilde a$ collect the exact and smoothed child values, $A$ the node's exact extremum, and $Q$ its operator of \cref{smd:eq:lower}. Monotonicity of $A$ and \cref{smd:prop:bands} give $A(a)\ge A(\widetilde a)\ge Q(\widetilde a)$. Since $A$ is nonexpansive in the infinity norm, $A(a)-A(\widetilde a)\le\max_uB_u$, and $A(\widetilde a)-Q(\widetilde a)\le b_{m_v}(\gamma_v)$ by \cref{smd:prop:bands}. Adding the two gives \cref{smd:eq:composition}, and unrolling \cref{smd:eq:budget} gives the path sum.
\end{proof}

The smoothed robustness is therefore a lower bound in the sense of \cref{smd:def:sound}, which is (R1). The budget is fixed by the graph and the parameters before any trajectory is evaluated. An until node is one maximum over $|I|$ witnesses whose child for witness $j$ is a minimum over $j+2$ values, so it adds two bands to a path for any window length,
\begin{equation}
B_{\mathbf U}=b_{|I|}(\gamma_{\mathrm{out}})+\max_{j\in I}\bigl[b_{j+2}(\gamma_j)+\max\{B_\phi,B_\psi\}\bigr].
\label{smd:eq:until-budget}
\end{equation}
The statements concern the sampled semantics in exact arithmetic. They do not cover the behavior between samples.

\ifarxiv
\begin{remark}[Cost of an until node]
\label{smd:rem:cost}
Because the inner minima of an until node share nested prefixes of $\rho^\phi$, one sort and a search tree over the entries~\cite{duchi2008efficient} give every inner value and the gradient of the node in $O(b\log b)$ operations for a window $I=\{a,\ldots,b\}$ (\cref{smd:app:cost}), against $O(b)$ for the exact minimum.
\end{remark}
\fi

\subsection{The log-sum-exp at an until node}
\label{smd:subsec:flip}
The LSE maximum $\overline L_\beta(z)=\beta^{-1}\log\sum_{i=1}^{m}e^{\beta z_i}$ lies above the maximum, and when $n$ entries attain the maximum,
\begin{equation}
\frac{\log n}{\beta}\le\overline L_\beta(z)-\max_iz_i\le\frac{\log m}{\beta}.
\label{smd:eq:plain}
\end{equation}
The LSE minimum $-\overline L_\beta(-z)$ lies below the minimum. We call a graph with these two operators the plain LSE. Its excess at a maximum grows with the number of tied entries, and an until operator produces such ties.

\begin{prop}[Sign flip at until]
\label{smd:prop:flip}
Let the exact value of an until node with window $I$ be $-\delta<0$, attained by $n$ of its witnesses. Replace the maximum over witnesses by $\overline L_\beta$ and every inner minimum by a lower approximation whose error on the attaining witnesses is at most $e$. Then the smoothed value $\widehat\rho$ of the node satisfies
\begin{equation}
\widehat\rho\ge-\delta-e+\frac{\log n}{\beta}.
\label{smd:eq:flip}
\end{equation}
Suppose that $\rho^\phi$ attains its minimum over the samples $t,\ldots,t+\max I$ before the window opens ($a\ge1$), that this minimum is $-\delta<0$, and that $\rho^\psi_{t+j}\ge-\delta$ for at least one $j\in I$. Then the exact value of the node is $-\delta$, and the attaining witnesses are exactly those $j$.
\end{prop}
\begin{proof}
The smoothed inner values of the $n$ attaining witnesses are at least $-\delta-e$. Since $\overline L_\beta$ is nondecreasing in every entry and its other entries contribute positive terms, $\widehat\rho$ is at least $\overline L_\beta$ of $n$ entries equal to $-\delta-e$, which is \cref{smd:eq:flip}. For the second statement, the prefix of every witness contains the minimizing sample, so its inner minimum is $\min\{\rho^\psi_{t+j},-\delta\}$, which equals $-\delta$ when $\rho^\psi_{t+j}\ge-\delta$ and is smaller otherwise.
\end{proof}

With $\beta=\log|I|/\varepsilon$, the plain LSE value of the violated node is positive whenever $\delta+e<\varepsilon\log n/\log|I|$. A trajectory whose deepest violation of $\phi$ up to the end of the window occurs before the window opens, and which satisfies $\psi$ throughout the window, ties all $|I|$ witnesses. The plain LSE value is then positive for every violation $\delta<\varepsilon-e$, so the constraint of an until conjunct in \cref{smd:eq:conjuncts} holds with $c=0$ on a violating trajectory. By \cref{smd:thm:composition}, the operators of \cref{smd:eq:lower} give it a negative value.

\subsection{Gradient allocation at equal error}
\label{smd:subsec:allocation}
Soundness alone is available from the LSE as well. Subtracting its worst-case excess gives the lower maximum
\begin{equation}
L_\beta(z)=\overline L_\beta(z)-\frac{\log m}{\beta},
\label{smd:eq:lse}
\end{equation}
which by \cref{smd:eq:plain} lies at most $\log m/\beta$ below the maximum and never above it. It is the generalized mean with generating function $e^{\beta x}$~\cite[Section~IV-C]{mehdipour2025generalized}. The LSE minimum has the same band below the minimum~\cite{gilpin2020smooth}. A graph with $L_\beta$ at maximum nodes and the LSE minimum at minimum nodes is therefore a lower bound by the argument of \cref{smd:thm:composition}. This graph is the sound LSE. Its gradient is the softmax of $\beta z$, which is positive on every entry. To compare the two families at one node of $m>1$ entries, we give them the same worst-case error $\varepsilon$,
\begin{equation}
\beta=\frac{\log m}{\varepsilon},\qquad
\gamma=\frac{2\varepsilon}{1-1/m}.
\label{smd:eq:matching}
\end{equation}

\begin{thm}[Gradient allocation]
\label{smd:thm:allocation}
Let a maximum node have $k\ge1$ children equal to $0$ and $N\ge1$ children equal to $-\Delta$, with $\Delta>0$ and $m=k+N$. The weights of $L_\beta$ and of $Q^+_\gamma$ on the $k$ children are
\begin{equation}
W_L=\frac{k}{k+N e^{-\beta\Delta}},\quad
W_Q=\min\Bigl\{1,\frac km+\frac{kN\Delta}{m\gamma}\Bigr\},
\label{smd:eq:mass}
\end{equation}
and $W_Q=1$ if and only if $\Delta\ge\gamma/k$. Under \cref{smd:eq:matching}, with $k$, $\varepsilon$, and $\Delta$ fixed and $N\to\infty$,
\begin{equation}
W_Q\to\min\bigl\{1,\tfrac{k\Delta}{2\varepsilon}\bigr\},
\ \ W_L\to\begin{cases}0,&\Delta<\varepsilon,\\k/(k+1),&\Delta=\varepsilon,\\1,&\Delta>\varepsilon.\end{cases}
\label{smd:eq:limits}
\end{equation}
The same statements hold at a minimum node after negating the children.
\end{thm}
\begin{proof}
$W_L$ is the sum of the softmax weights of the $k$ children. For $Q^+_\gamma$, if the $N$ children are above the threshold, symmetry and \cref{smd:eq:threshold} give $\theta=-(N\Delta+\gamma)/m$, so each of the $k$ children has weight $(N\Delta+\gamma)/(m\gamma)$, and $\theta<-\Delta$ holds precisely when $k\Delta<\gamma$. Otherwise the $N$ children receive zero and the $k$ children $1/k$ each, which proves \cref{smd:eq:mass}. Under \cref{smd:eq:matching}, $\gamma\to2\varepsilon$ and $N/m\to1$, which gives the limit of $W_Q$, and $Ne^{-\beta\Delta}=(N/m)\,m^{1-\Delta/\varepsilon}$ diverges, tends to one, or vanishes as $\Delta$ is below, at, or above $\varepsilon$.
\end{proof}

\begin{corollary}[Weight under a gap]
\label{smd:cor:gap}
Let the $k$ largest children of a maximum node exceed each of its $N$ other children by at least $\Delta>0$. Then, for every $N$ and all values of the children, $W_Q\ge\min\{1,k\Delta/\gamma\}$ and $W_L\ge k/(k+Ne^{-\beta\Delta})$. The configuration of \cref{smd:thm:allocation} attains the bound on $W_L$ for every $N$.
\end{corollary}
\begin{proof}
Let $c$ be the smallest of the $k$ children and $\theta$ the threshold of \cref{smd:eq:threshold}. If $c-\theta\le\Delta$, every other child lies at or below $\theta$ and $W_Q=1$. Otherwise each of the $k$ children has weight at least $(c-\theta)/\gamma>\Delta/\gamma$. The bound on $W_L$ holds because the softmax weight of the $k$ children increases with their values and decreases with those of the others.
\end{proof}

\begin{remark}[Unequal children]
\label{smd:rem:unequal}
The $k$ children need not be equal. Let them take values $-\delta_i\le0$ with $\min_i\delta_i=0$, above the $N$ others at $-\Delta$. The softmax sum gives $W_L=k_{\mathrm{eff}}/(k_{\mathrm{eff}}+Ne^{-\beta\Delta})$ with the effective count $k_{\mathrm{eff}}=\sum_ie^{-\beta\delta_i}\le k$, where $Ne^{-\beta\Delta}$ becomes $\sum_je^{-\beta\Delta_j}$ if the others lie at $-\Delta_j$, and \cref{smd:eq:threshold} gives $W_Q=1$ if and only if $\sum_i(\Delta-\delta_i)\ge\gamma$. Both statements reduce to \cref{smd:eq:mass} for equal children.
\end{remark}

\Cref{smd:thm:allocation,smd:cor:gap} answer (R2): the sparsemax weight on the $k$ children is bounded below independently of $N$ and of the values of the other children, is the whole weight once $\Delta\ge\gamma/k$, and, since $\gamma\le4\varepsilon$ under \cref{smd:eq:matching} for every $m\ge2$, is at least $\min\{1,k\Delta/(4\varepsilon)\}$ at every $N$. For $\Delta<\varepsilon$ the sound LSE has no such bound, since its weight in the configuration of \cref{smd:thm:allocation} is the smallest possible and decays as $N^{-(1-\Delta/\varepsilon)}$ (\cref{smd:fig:operator}(b,c)). For $\Delta>\varepsilon$ the weight of the sound LSE on the $k$ children also tends to 1, so choosing $\varepsilon$ below the gap restores its weight. The gap is a property of the trajectory, however, and is not known when $\varepsilon$ is chosen. Sparsemax keeps the whole weight for every $\varepsilon\le k\Delta(1-1/m)/2$, at every $N$. This range exceeds the range $\varepsilon<\Delta$ of the sound LSE by the factor $k(1-1/m)/2$, which is above one for $k\ge3$. A smaller $\gamma$ widens that range of gaps, but the gradient of the lower extrema is Lipschitz in the Euclidean norm with constant $1/\gamma$, so the constraint becomes less smooth, and as $\gamma\to0$ the operator tends to the exact extremum; at equal $\varepsilon$ this constant, $(1-1/m)/(2\varepsilon)$, is below the LSE gradient's $\log m/(2\varepsilon)$ for every $m\ge2$. The matching \cref{smd:eq:matching} equalizes the worst-case error and not the realized error. At a tie of $k$ entries of a minimum node whose other entries lie far above the tie, the error of $Q^-_\gamma$ is $\varepsilon(1-1/k)/(1-1/m)$ and that of the LSE minimum tends to $\varepsilon\log k/\log m$, which is smaller for $1<k<m$, so sparsemax gives the smaller value at equal $\varepsilon$.

The inner minimum of an until node runs over every sample of the prefix, so a trajectory that violates $\phi$ at $k$ samples and then keeps $\rho^\phi$ a gap $\Delta$ above the violation for $N$ samples gives the node the structure of \cref{smd:thm:allocation}, with $N$ growing by one per sampling period of the wait. Through the chain rule the controls receive $\sum_ip_{\gamma,i}\nabla_Uz_i$, so the $N$ samples contribute nothing under sparsemax once $\Delta\ge\gamma/k$ and almost everything under the sound LSE once $W_L$ is small. Whether this weight moves the trajectory out of the violation depends on the dynamics (\cref{smd:sec:results}).

\section{Numerical Examples}
\label{smd:sec:results}
\subsection{Planar unicycle}
\label{smd:subsec:planar}
A unicycle with position $(x,y)$ and heading $\theta$ follows $x^+=x+hv\cos\theta$, $y^+=y+hv\sin\theta$, $\theta^+=\theta+h\omega$ with $h=0.1$ s, speed $|v|\le1$ m/s, and turn rate $|\omega|\le\pi/2$ rad/s. Four disks lie in a square workspace of side 10 m (\cref{smd:fig:planar}(a)), and the predicate of a disk is the signed distance to its boundary in metres. The unicycle must stay out of Red until it is in Green within a window of 2 s, then reach Blue and stay there for 2 s, never enter the Obstacle, and stay in the workspace,
\begin{equation}
\begin{aligned}
\phi={}&(\neg\mathrm{Red}\ \mathbf U_{I_1}\ \mathrm{Green})\ \wedge\ \mathbf F_{I_2}\mathbf G_{[0,2]}\mathrm{Blue}\\
&\wedge\ \mathbf G\,\neg\mathrm{Obstacle}\ \wedge\ \mathbf G\,\mathrm{Boundary}.
\end{aligned}
\label{smd:eq:planar_spec}
\end{equation}
Two fixed trajectories round Red on the way to Green. $S_1$ cuts into Red on an arc concentric with it, so $k=21$ samples lie inside Red at a depth of at most $\delta=0.015$ m, and then waits 0.023 m outside Red until the window $I_1$ opens. We lengthen the wait from 3 to 130 s, which keeps $N=31$ to 1301 samples a gap $\Delta=0.038$ m above the violation. $S_2$ passes Red at a clearance of 0.3 m at one point and satisfies $\phi$ with exact robustness 0.3, which the until conjunct sets.

We compare five robustness measures. The plain LSE, the sound LSE, and sparsemax have the same worst-case error $\varepsilon=0.1$ m per node under \cref{smd:eq:matching}, so $\delta=0.15\varepsilon$ and $\Delta=0.38\varepsilon$. The generalized-mean robustness (GMR)~\cite{mehdipour2025generalized} replaces a conjunction of $d$ operands by the power mean $(\frac1d\sum_ix_i^p)^{1/p}$, the geometric mean at $p=0$, when every operand is positive and by $-(\frac1d\sum_i|\min\{x_i,0\}|^q)^{1/q}$ otherwise. It evaluates each witness of until as a two-operand conjunction of the right operand with the conjunction over the prefix. It is sound for all $p$ and $q$~\cite[Theorem~2]{mehdipour2025generalized}, since its sign is that of the exact robustness, and it preserves the sign and not a margin (\cref{smd:sec:related}). We use the orders $(p,q)=(0,1)$, a setting used in~\cite{mehdipour2025generalized}, and $(-10,10)$, whose power means lie closer to the extrema. As $p\to-\infty$ and $q\to\infty$ GMR tends to the exact robustness~\cite[Sections~IV-B and~IV-C]{mehdipour2025generalized}, as the LSE and sparsemax do for $\varepsilon\to0$, so its order sets how much the other samples contribute against its closeness to the exact robustness, as $\varepsilon$ does.

\begin{table}[t]
\caption{Planar unicycle at $\varepsilon=0.1$ m: values of $\phi$ on the fixed trajectories at the 3 s wait, weight on the samples of $S_1$ inside Red at waits of 3 and 130 s, and exact robustness of the trajectory returned from $S_1$ at constraint level $c$. Entries shown as $<0.001$ lie in $(0,0.001)$.}
\label{smd:tab:planar}
\centering
\footnotesize
\setlength{\tabcolsep}{2.4pt}
\begin{tabular}{@{}lcccccc@{}}
\toprule
& \multicolumn{2}{c}{value} & \multicolumn{2}{c}{weight on $S_1$} & \multicolumn{2}{c}{returned $\rho$} \\
\cmidrule(lr){2-3}\cmidrule(lr){4-5}\cmidrule(lr){6-7}
measure & $S_1$ & $S_2$ & 3 s & 130 s & $c=0$ & $c=0.2$ \\
\midrule
exact & $-0.015$ & 0.300 & & & & \\
plain LSE & $+0.024$ & 0.357 & 0.68 & 0.19 & $-0.058$ & 0.146 \\
sound LSE & $-0.076$ & 0.258 & 0.68 & 0.19 & 0.045 & 0.218 \\
GMR $(0,1)$ & $-0.0001$ & 1.69 & 0.043 & 0.007 & $<0.001$ & $<0.001$ \\
GMR $(-10,10)$ & $-0.009$ & 0.49 & 0.71 & 0.60 & $<0.001$ & 0.141 \\
sparsemax & $-0.110$ & 0.214 & 1.000 & 1.000 & 0.087 & 0.248 \\
\bottomrule
\end{tabular}
\end{table}

\begin{figure}[t]
\centering
\includegraphics{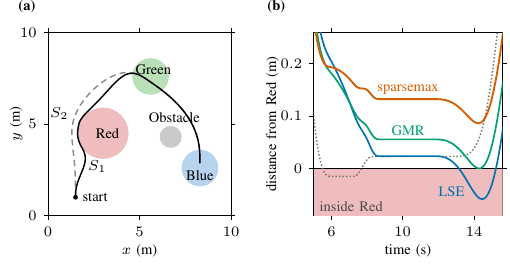}
\caption{Planar unicycle. (a) The regions and the two fixed trajectories: $S_1$ cuts into Red and waits beside it, $S_2$ passes Red at a clearance of 0.3 m. (b) Distance from Red of the trajectories returned at $c=0$ under the plain LSE (solid) and the sound LSE (dashed), GMR of order $(-10,10)$, and sparsemax, from the initial trajectory $S_1$ (dotted).}
\label{smd:fig:planar}
\end{figure}

On $S_1$ the unicycle is in Green through the whole window, so all 21 witnesses tie and the plain LSE is positive (\cref{smd:tab:planar}) on a trajectory 0.015 m inside Red, within $0.006\varepsilon$ of the bound of \cref{smd:prop:flip} with an inner-minimum error of $0.62\varepsilon$; the error of sparsemax, $0.95\varepsilon$, arises at the inner minimum over the tied samples (\cref{smd:prop:bands}), and both orders of GMR have the correct sign. On $S_2$, at a clearance of 0.3 m, the plain LSE and GMR exceed it, whereas the sound LSE and sparsemax give margins that \cref{smd:thm:composition} certifies. The weight of sparsemax on the samples inside Red stays 1, that of both LSE variants falls as \cref{smd:thm:allocation} predicts for $\Delta<\varepsilon$, and that of GMR scales, by eqs.~(12) and~(15) of~\cite{mehdipour2025generalized}, as $(k/(2d))^{1/q}$ with the $d$ samples of the prefix.

We then minimize the effort $E=\frac1T\sum_t\|u_t/u_{\max}\|_2^2$ subject to \cref{smd:eq:conjuncts}, from $S_1$, with a first-order augmented Lagrangian method under single shooting over the box-bounded controls: 400 updates, each a step along the normalized gradient of the augmented Lagrangian under an Armijo backtracking, with the multiplier and penalty of each conjunct revised every ten updates. Every smoothed constraint then holds within 0.0003 m, and the exact robustness values that follow are measured on the returned trajectories. With $c=0$ every constraint of the plain LSE holds on a trajectory that cuts 0.058 m into Red (\cref{smd:fig:planar}(b)). The constraint of GMR at zero is the exact constraint, and the trajectory of order $(-10,10)$ touches Red. The sound LSE and sparsemax keep the unicycle 0.045 and 0.087 m outside. With $c=0.2$ m only these two lower bounds return a trajectory whose exact robustness is at least $c$, at efforts 1\% and 4\% above that of the plain LSE. From $S_2$ the plain LSE returns a violating trajectory with $c=0$, and the two lower bounds meet $c=0.2$ m. The authors of GMR instead maximize it in the cost under a margin constraint~\cite{mehdipour2025generalized}.

\subsection{Manipulator beside a person}
\label{smd:subsec:manipulator}
\begin{figure*}[t]
\centering
\includegraphics[width=0.8\textwidth,trim=0 298 0 0,clip]{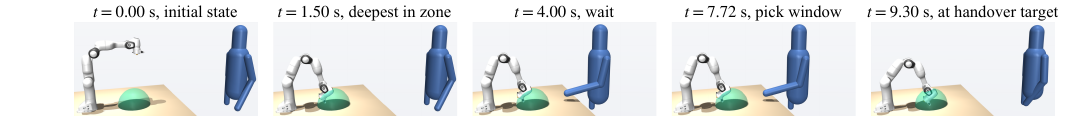}\\[3pt]
\includegraphics[width=0.8\textwidth,trim=0 0 0 171,clip]{scene_collage_person.pdf}
\caption{A violating initial trajectory of the manipulator ($H=10$ s, exact robustness $-0.05$). Top: manipulator, person and zone (green) at five times. Bottom left: exact predicate values; the end effector is inside the zone for 0.3 s and then waits at the pick pose. Bottom right: weight on the samples inside the zone at $\varepsilon=0.2$, medians over the four trajectories (markers), with the curves of \cref{smd:rem:unequal} for the LSE and sparsemax and straight lines for GMR.}
\label{smd:fig:scene}
\end{figure*}

A simulated seven-joint Franka Panda manipulator under torque control shares a table with a person. It picks a part outside a zone, a ball of radius 0.2 m centred 0.5 m in front of its base and reserved for the person's reach, and hands it over inside the zone (\cref{smd:fig:scene}). Each predicate is a distance or speed margin divided by its threshold. The length is a smooth norm, $\sqrt{\|x\|^2+\eta^2}$ in a margin of the form threshold minus length and $\sqrt{\|x\|^2+\eta^2}-\eta$ in one of the form length minus threshold, with $\eta$ one percent of the smallest threshold, so every leaf lies at most $e_v=0.01$ below its exact value and never above it. With times in seconds, a horizon $H$, a pick window $[w,w+1]$, and a handover window $I_{\mathrm h}$ after it, the specification is
\begin{align}
\phi={}&\mathbf G_{[0,H]}\bigwedge_{\ell,s}\mathrm{sep}_{\ell s}
\wedge\mathbf G_{[0,H]}\bigwedge_{\ell}\Bigl(\bigwedge_{s}\mathrm{slow}_{\ell s}\vee\mathrm{speed}_\ell\Bigr)\nonumber\\
&\wedge\bigl(\mathrm{clear}\ \mathbf U_{[w,w+1]}\ \mathrm{pick}\bigr)
\wedge\mathbf F_{I_{\mathrm h}}\mathbf G_{[0,0.8]}\mathrm{hand},
\label{smd:eq:robot_spec}
\end{align}
where $\mathrm{sep}$ and $\mathrm{slow}$ bound the separation of the manipulator's spheres $\ell$ from the person's spheres $s$, $\mathrm{speed}$ bounds their speed, $\mathrm{clear}$ is the distance of the end effector from the zone, and $\mathrm{pick}$ and $\mathrm{hand}$ are tolerances about two poses. The thresholds are a separation of 0.10 m for $\mathrm{sep}$, 0.30 m and 0.25 m/s for $\mathrm{slow}$ and $\mathrm{speed}$, and 0.05 m about the pick and handover targets; the pick window opens at $w=H-2.78$ s and the handover window $I_{\mathrm h}$ follows it and its dwell ends at $H$. We sample at $h=0.02$ s, simulate with the Warp implementation of MuJoCo~\cite{todorov2012mujoco}, and differentiate it in reverse mode.

On four initial trajectories, which differ in one intermediate configuration, the end effector overshoots into the zone on its way to the pick pose, and the manipulator then waits at that pose until the pick window closes at $H-1.78$ s. In units of the worst-case error $\varepsilon=0.2$ per node, the overshoot is $0.25\varepsilon$ deep and the pick pose lies a gap $\Delta=0.375\varepsilon$ above it. The end effector is inside the zone at $k=18$ or 19 samples, and the wait contains $N=287$, 787, 1787, and 3787 samples at horizons of 10, 20, 40, and 80 s. Each trajectory violates the until conjunct with exact robustness $-0.05$ and satisfies the other three conjuncts.

\begin{table}[t]
\caption{The until conjunct on the violating initial trajectories of the manipulator at $\varepsilon=0.2$, medians over the four trajectories: value at $H=10$ s, sum of the leaf weights on the samples inside the zone, and norm of the part of the torque gradient that passes through those samples.}
\label{smd:tab:manipulator}
\centering
\footnotesize
\setlength{\tabcolsep}{3pt}
\begin{tabular}{@{}lccccc@{}}
\toprule
& & \multicolumn{2}{c}{weight on} & \multicolumn{2}{c}{torque gradient through} \\
& & \multicolumn{4}{c}{the samples inside the zone} \\
\cmidrule(lr){3-4}\cmidrule(lr){5-6}
measure & value & 10 s & 80 s & 10 s & 80 s \\
\midrule
exact & $-0.050$ & & & & \\
plain LSE & $+0.023$ & 0.216 & 0.036 & 3.62 & 0.67 \\
sound LSE & $-0.177$ & 0.216 & 0.036 & 3.62 & 0.67 \\
GMR $(0,1)$ & $-0.0006$ & 0.024 & 0.0024 & 0.40 & 0.039 \\
GMR $(-10,10)$ & $-0.029$ & 0.588 & 0.467 & 10.1 & 8.0 \\
sparsemax & $-0.220$ & 1.000 & 1.000 & 16.82 & 16.82 \\
\bottomrule
\end{tabular}
\end{table}

The manipulator waits at the pick pose through the whole pick window, so all 51 witnesses tie and the plain LSE is positive on the violated conjunct (\cref{smd:tab:manipulator}), as \cref{smd:prop:flip} predicts. The sound LSE and sparsemax give values below the exact one and within the budget $B_{\mathbf U}=2\varepsilon+e_v=0.41$ of \cref{smd:eq:until-budget}. The weight of sparsemax on the samples inside the zone is 1.000 at every horizon, because every sample of the wait lies beyond the threshold $\theta$ of \cref{smd:eq:threshold}. The weight of the LSE falls from 0.216 to 0.036 and stays within 5\% of the expression of \cref{smd:rem:unequal} (\cref{smd:fig:scene}). By eqs.~(12) and~(15) of~\cite{mehdipour2025generalized}, the weight of GMR scales as $(k/(2d))^{1/q}$. The chain rule carries these weights to the torques. The part of the gradient with respect to the normalized torques $u_t/u_{\max}$ that passes through the samples inside the zone, in the Euclidean norm over all samples, is the same at 10 and 80 s under sparsemax and falls by a factor of 5.4 under the LSE and of 10 under GMR of order $(0,1)$.

We repeat the optimization on the torques with $c=0$ from initial trajectories with exact robustness $-0.05$ and $-0.10$, overshoots of 1.0 and 2.0 cm, at 16 pick and handover locations and horizons of 10, 20, and 40 s with the solver of the unicycle example, whose normalized step removes the gradient's size from the update. In the median over the 16 locations, the first exactly satisfying trajectory comes at update 1 or 2 under sparsemax and GMR of order $(-10,10)$ at every horizon and from 4.5 to 19.5 updates, growing with the horizon, under the sound LSE (\cref{smd:tab:outcomes}). Under the solver's own constraint the medians are 10 to 20 updates for sparsemax and 12.5 to 29 for the sound LSE, and 3 of 32 sparsemax trajectories at 40 s never meet it, since its realized error is larger (\cref{smd:tab:manipulator}). After 100 updates at 40 s the median exact robustness is near zero under the plain LSE, whose constraints hold at the violating initial trajectory so its first updates reduce the effort alone, and under GMR, which preserves the sign and not a margin, while the two lower bounds keep their realized errors as margins.

\begin{table}[t]
\caption{Optimization on the torques from the violating initial trajectories, the same 16 locations for every measure, $c=0$, 100 updates: median first update with a nonnegative exact robustness at 10, 20, and 40 s, with quartiles at 40 s; at 40 s, the median and smallest exact robustness after 100 updates and the median effort relative to the initial trajectory, $E_{100}/E_0$. A superscript counts the locations that never satisfied, left out of the median.}
\label{smd:tab:outcomes}
\centering
\footnotesize
\setlength{\tabcolsep}{1.5pt}
\begin{tabular}{@{}llccccc@{}}
\toprule
& & \multicolumn{3}{c}{first satisfying update} & \multicolumn{2}{c}{after 100 updates, 40 s} \\
\cmidrule(lr){3-5}\cmidrule(lr){6-7}
measure & $\rho_0$ & 10 s & 20 s & 40 s & $\rho_{100}$ & effort \\
\midrule
plain LSE & $-0.05$ & 6 & 6 & 6 [5,7] & 0.001 ($-0.006$) & 0.61 \\
 & $-0.10$ & $7^{1}$ & $7^{2}$ & 7 [7,23] & $-0.0005$ ($-0.15$) & 0.50 \\
sound LSE & $-0.05$ & 4.5 & 10 & 17 [14,22] & 0.097 (0.06) & 0.94 \\
 & $-0.10$ & 6 & 11 & 19.5 [13,25] & 0.083 (0.04) & 0.94 \\
GMR $(0,1)$ & $-0.05$ & 3 & $5^{1}$ & $5^{1}$ [5,6] & 0.003 ($-0.02$) & 0.33 \\
 & $-0.10$ & $6^{1}$ & $11^{3}$ & $7^{4}$ [5,15] & 0.005 ($-0.03$) & 0.33 \\
GMR $(-10,10)$ & $-0.05$ & 1 & 1 & 1 [1,2] & 0.014 ($-0.003$) & 0.65 \\
 & $-0.10$ & 2 & 2 & 2 [2,3] & 0.009 ($-0.005$) & 0.67 \\
sparsemax & $-0.05$ & 1 & 1 & 1 [1,2] & 0.172 (0.16) & 0.97 \\
 & $-0.10$ & 2 & 2 & 2 [2,4] & 0.166 (0.15) & 0.98 \\
\bottomrule
\end{tabular}
\end{table}

\section{Conclusion}
\label{smd:sec:conclusion}
The sparsemax lower extrema are lower bounds with an error per node of at most $\gamma/2$ for any number of samples, and their weight on $k$ samples a gap $\Delta$ beyond all others is at least $\min\{1,k\Delta/\gamma\}$ for every horizon. At equal worst-case error $\varepsilon$ this weight is at least $\min\{1,k\Delta/(4\varepsilon)\}$, while the weight of the sound LSE tends to zero as the number of other samples grows when $\Delta<\varepsilon$ (\cref{smd:thm:allocation,smd:cor:gap}).

Sparsemax has costs. Its minimum's band is attained at a tie, so at equal worst-case error sparsemax gives the smaller value (\cref{smd:subsec:allocation}) and rejects every trajectory whose exact margin lies below its realized error, 0.17 against 0.13 for the sound LSE in \cref{smd:tab:manipulator}, which can leave no feasible plan where feasibility is tight. A sample beyond the threshold receives no derivative, so nothing warns the solver of one about to become critical, and in a nested specification a zero at an inner node hides that sample from the nodes above it. The operators are once continuously differentiable, and the update advantage we measured is a first-order one. GMR preserves the sign of the exact robustness and not a margin, so under an effort cost with a threshold of zero its trajectories end at an exact robustness near zero, where a lower bound keeps its realized error as a margin.\ifarxiv{} \Cref{smd:tab:stops} reports these measurements.\fi{} For future work, we seek to exploit the sparse Hessian of sparsemax when the sampled states are decision variables, through sparse automatic differentiation for second-order solvers.

%
\section*{Acknowledgment}
    The authors acknowledge the use of Claude LLMs~\cite{anthropic2026claude} for code development, visualizations, and manuscript review.

\ifarxiv
\input{sections/A-appendix}
\fi

\bibliographystyle{IEEEtran}
\bibliography{references}

\end{document}